\documentclass[conference]{IEEEtran}
\usepackage{CJKutf8}
\usepackage{changes}
\usepackage{balance}

\usepackage{gensymb}
\usepackage{bm}

\usepackage{amsmath,amssymb}
\usepackage{amsthm}
\usepackage{marvosym}

\usepackage{graphicx}
\graphicspath{{figures/}}
\usepackage{subcaption}
\usepackage[export]{adjustbox}
\usepackage{flushend}

\usepackage{listings}
\usepackage{color}
\usepackage{soul}
\usepackage{tikz}
\newcommand*{\circled}[1]{\lower.7ex\hbox{\tikz\draw (0pt, 0pt)%
    circle (.5em) node {\makebox[1em][c]{\small #1}};}}

\usepackage{algorithm}
\usepackage{algpseudocode}

\theoremstyle{plain}
\newtheorem{theorem}{Theorem}[section]
\newtheorem{corollary}{Corollary}[section]
\newtheorem{lemma}{Lemma}[section]

\newtheorem{example}{Example}[section]
\newtheorem{proof}{Proof}

\theoremstyle{definition}
\newtheorem{definition}{Definition}[section]
 
\theoremstyle{remark}

\definecolor{mygreen}{rgb}{0,0.6,0}
\definecolor{mygray}{rgb}{0.5,0.5,0.5}
\definecolor{mymauve}{rgb}{0.58,0,0.82}

\usepackage{algpseudocode}
\usepackage{algorithm}

\usepackage{url}
\usepackage{booktabs} 
\usepackage{multirow}
\usepackage{array}
\newcolumntype{I}{!{\vrule width 1.2pt}}
\newlength\savedwidth

\newlength\savewidth
\newcommand\shline{\noalign{\global\savewidth\arrayrulewidth
                           \global\arrayrulewidth 1.2pt}%
                  \hline
                  \noalign{\global\arrayrulewidth\savewidth}}
\usepackage{makecell}

\usepackage{geometry}
\begin{document}
\begin{CJK*}{UTF8}{gbsn}

\title{Enhanced Fast Boolean Matching based on \\Sensitivity Signatures Pruning}

\author{\IEEEauthorblockN{
Jiaxi~Zhang$^{1,2,}$\IEEEauthorrefmark{1},
Liwei~Ni$^2$,
Shenggen~Zheng$^{2}$, 
Hao~Liu$^{1}$,
Xiangfu~Zou$^2$,
Feng~Wang$^1$
and~Guojie~Luo$^{1,2,3,\text{\Letter}}$
}
\IEEEauthorblockA{
$^1$Center for Energy-Efficient Computing and Applications, Peking University, Beijing, China}
\IEEEauthorblockA{
$^2$Peng Cheng Laboratory, Shenzhen, China
}
\IEEEauthorblockA{
$^3$Advanced Institute of Information Technology, Peking University, Hangzhou, China
}
\IEEEauthorblockA{
Email: \{\IEEEauthorrefmark{1}zhangjiaxi, $^{\text{\Letter}}$gluo\}@pku.edu.cn
}}

\maketitle

\balance

\begin{abstract}
Boolean matching is significant to digital integrated circuits design.
An exhaustive method for Boolean matching is computationally expensive even for functions with only a few variables, because the time complexity of such an algorithm for an n-variable Boolean function is $O(2^{n+1}n!)$.
Sensitivity is an important characteristic and a measure of the complexity of Boolean functions.
It has been used in analysis of the complexity of algorithms in different fields.
This measure could be regarded as a signature of Boolean functions and has great potential to help reduce the search space of Boolean matching.

In this paper, we introduce Boolean sensitivity into Boolean matching and design several sensitivity-related signatures to enhance fast Boolean matching.
First, we propose some new signatures that relate sensitivity to Boolean equivalence.
Then, we prove that these signatures are prerequisites for Boolean matching, which we can use to reduce the search space of the matching problem.
Besides, we develop a fast sensitivity calculation method to compute and compare these signatures of two Boolean functions.
Compared with the traditional cofactor and symmetric detection methods, sensitivity is a series of signatures of another dimension.
We also show that sensitivity can be easily integrated into traditional methods and distinguish the mismatched Boolean functions faster.
To the best of our knowledge, this is the first work that introduces sensitivity to Boolean matching.
The experimental results show that sensitivity-related signatures we proposed in this paper can reduce the search space to a very large extent, and perform up to 3x speedup over the state-of-the-art Boolean matching methods.

\end{abstract}


\section{Introduction}\label{sec:intro}

Boolean equivalence classification and matching are widely used in many design stages such as logic synthesis, engineering change order, verification, and hardware Trojan detection.
A key task of Boolean matching is to determine whether two Boolean functions belong to the same NPN class.
An NPN class is a set of completely Boolean functions, all of which can be obtained from each other with three types of transformations including permuting the inputs or complementing the inputs and outputs.
There are $2^{n+1}n!$ NPN transformations for an n-variable Boolean function.
An exhaustive method can determine whether two Boolean functions are equivalent by enumerating these transformations, but the running time will be unacceptable as $n$ increases.

Boolean matching is a long-term problem due to its huge computational complexity.
Many methods have been explored to solve this problem.
These methods usually take truth tables or binary decision diagrams~(BDDs) as the inputs of matching.
These works can be classified as four types~\cite{benini1997survey}, algorithms based on canonical forms, algorithms using Boolean signatures, SAT-based methods, and spectral-analysis-based methods. 
Algorithms based on canonical formwork by computing some complete and unique
canonical forms of the Boolean functions, and all Boolean functions in an equivalence class have the same canonical form.
This form can be used to check for NPN equivalence by straightforwardly testing NPN transformations.
Signatures of a Boolean function, which also called filters, are compact representations that characterize some of the properties of the function itself. 
The search space was reduced and the matching speed was improved by means of structural signatures.
Spectral-based methods usually transform Boolean function into spectral representations, where a representation can uniquely identify a function. 
SAT-based methods rely on quick SAT solvers.
These methods usually derive the SAT formulation based on the specific application of Boolean matching.
In a word, it is hard to directly test the NPN equivalent by applying NPN transformations.
The key point of Boolean matching is to find inherent properties of Boolean functions to prune and reduce the search space.

Sensitivity was first introduced~\cite{cook1982bounds} as a simple combinatorial complexity measure for Boolean functions.
It is nowadays a well-known invariant of Boolean functions that occurs in many different fields, ranging from satisfied problem~\cite{dubois1997general, kirousis1998approximating, impagliazzo2001complexity} to quantum computational complexity~\cite{buhrman2002complexity}.
The sensitivity set of a Boolean function at a particular input is the set of input positions where changing that one bit then the output will be changed.
The sensitivity of the Boolean function at a particular input is then the cardinality of the sensitivity set, while the sensitivity of the function is defined as the maximum of its sensitivity over all possible inputs.

Sensitivity can be regarded as a series of signatures of the Boolean functions.
This series of signatures also includes block sensitivity~\cite{nisan1991crew}, average sensitivity, and average block sensitivity.
Amano~\cite{amano2017enumeration} gave some statistical data on sensitivity and NPN equivalence classes.
In fact, Boolean functions with different sensitivity properties could not be NPN equivalent~(see detailed proofs in Section~\ref{sec:prove}).
This feature gives sensitivity great potential to help reduce the search space of Boolean matching.
Previous signatures are mainly based on cofactor and symmetries of Boolean functions, they only explore more about symmetric variables of Boolean functions.
Sensitivity contains more structured information between variables~(see details in Section~\ref{sec:compare}).
In this paper, we will consider the sensitivity of Boolean functions and propose several techniques to enhance fast Boolean matching based on series of sensitivity signatures.
Our contributions in this paper are fourfold:
\begin{itemize}
    \item To the best of our knowledge, this is the first work that introduces sensitivity into Boolean matching. We propose some new signatures that relate sensitivity to Boolean equivalence.
    \item We prove that these signatures are prerequisites for Boolean matching, which we can use to reduce the search space of the matching problem.
    Experimental results show that sensitivity signatures have a high pruning effect.
    \item We develop a fast sensitivity calculation method to compute and compare sensitivity-related signatures of two Boolean functions. This method can quickly determine whether the sensitivities of two Boolean functions are equal.
    \item We show that sensitivity can be easily integrated into traditional methods and distinguish the mismatched Boolean functions faster. Experimental results show that the overall method can perform up to 3x speedup over the state-of-the-art Boolean matching methods.
\end{itemize}

The rest of the paper is organized as follows.
Section~\ref{sec:background} summarizes the background of Boolean matching and Boolean sensitivity.
Section~\ref{sec:prove} provides some definitions of sensitivity- related signatures, and some theorems and their proofs used in Boolean matching.
Section~\ref{sec:method} explains how we can use series of sensitivity signatures to enhance the fast Boolean matching method.
Implementation are evaluated with experimental results in Section~\ref{sec:evaluation}.
Finally, Section~\ref{sec:related} introduces some related works, and Section~\ref{sec:conclusion} concludes the paper.

\section{Preliminaries}\label{sec:background}

\subsection{Notations and Basic Definitions}

An $n$-variable Boolean function $f(x)$ takes the form $f: \{0,1\}^n \rightarrow \{0,1\}$, where $\{0,1\}$ is the Boolean domain and $n$ is the arity of $f$.
We call $x \in \{0,1\}^n$ a \emph{word} of arity $n$,
and We denote the $i$-th bit in the word as $x_i$.
Thus, $x=(x_1,x_2,\cdots,x_n)$ is also a Boolean string of length $n$.
In this paper, we use $f$ and $g$ to denote Boolean functions on $n$ variables.
Unless otherwise stated, $x, y, z$ denote words of arity $n$. 

\begin{figure}[tbp]
\centering
\begin{subfigure}[b]{0.45\linewidth}
\centering
\includegraphics[width=0.90\linewidth]{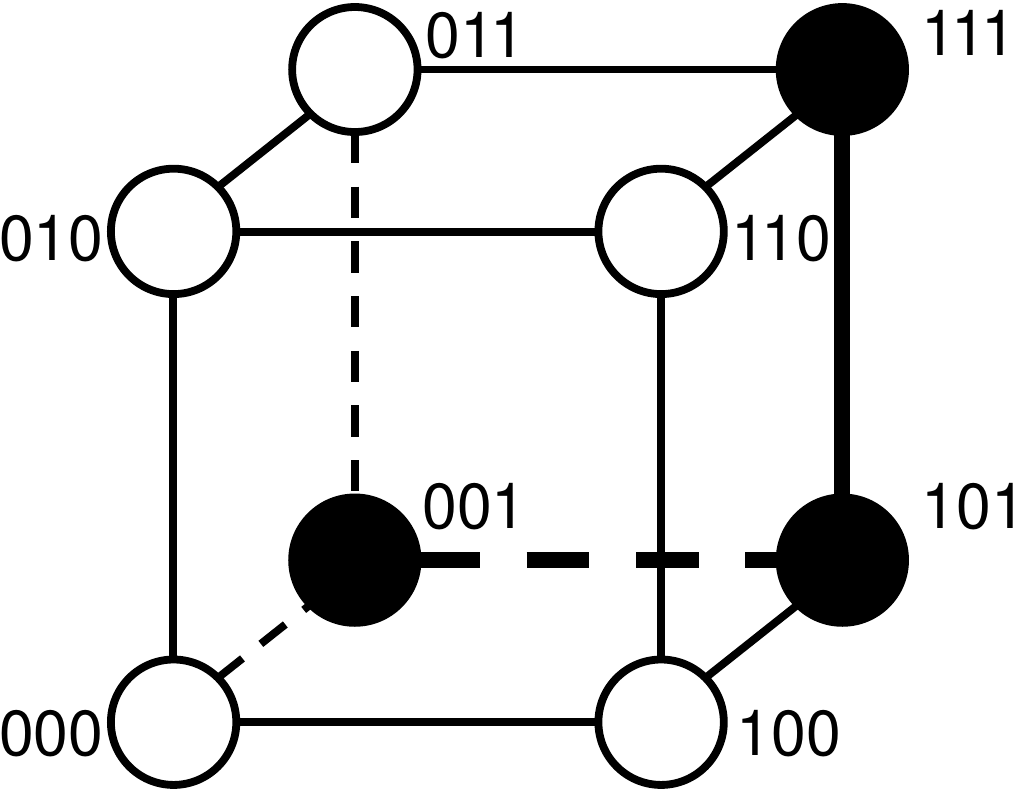}
\caption{Hypercube view}
\label{fig:hypercube1}
\end{subfigure}
\begin{subfigure}[b]{0.45\linewidth}
\centering
\includegraphics[width=0.90\linewidth]{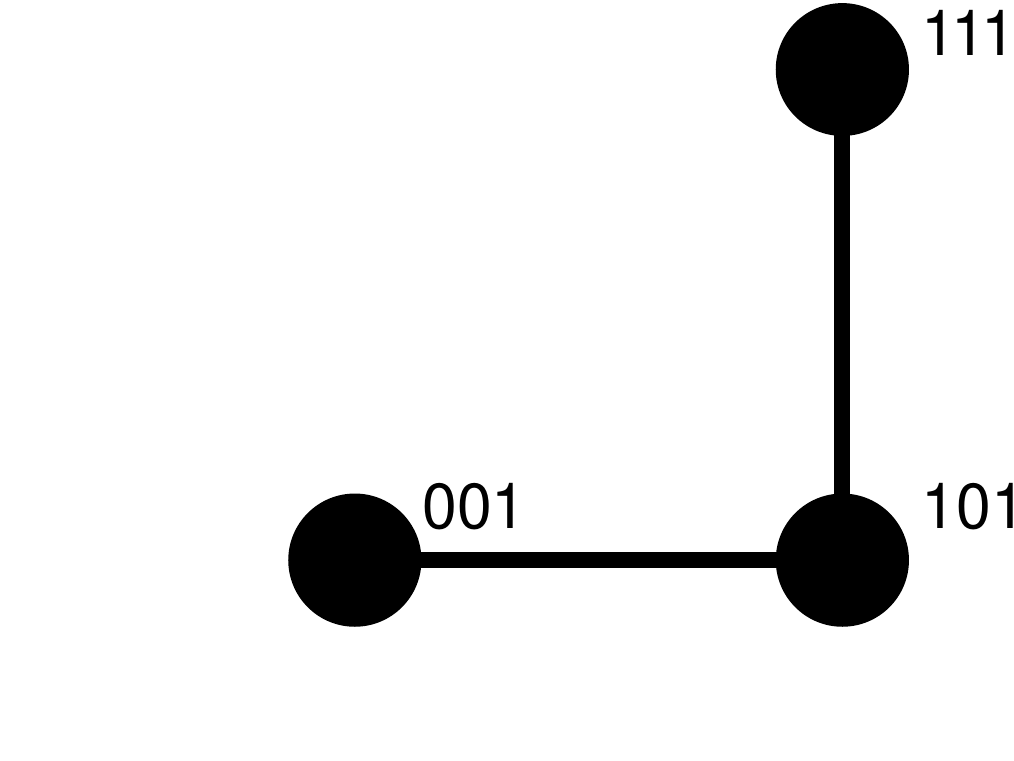}
\caption{Induced subgraph}
\label{fig:hypercube2}
\end{subfigure}
\caption{Graph representation for $f=x_1x_2x_3+\overline{x_2}x_3$.}
\label{fig:hypercube}
\vspace{-1em}
\end{figure}

Boolean function $f$ is often represented by its truth table $T(f)$, a string of $2^n$ bits.
The $i$-th bit of $T(f)$ is equal to $f((i)_2)$, where $(i)_2$ is the little-endian binary code of integer $i$.
From $T(f)$, we can express $f$ as a sum of 1-minterm.
We denote $X$ as the number of terms in truth table and $x^{(i)}$ as the $i$-th term. 

A Boolean function can also be represented by a subgraph of a hypercube.
The hypercube $Q_n$ is a graph of order $2^n$ whose vertices are represented by $n$-tuples $(x_1,x_2,...,x_n)$, where $x_i\in \{0,1\}$, and whose edges connect vertices which differ in exactly one term.
$f$ can be represented as the induced subgraph of $Q_n$ from the 1-minterm nodes.
Figure~\ref{fig:hypercube} gives an example.
Figure~\ref{fig:hypercube2} is induced subgraph from $Q_3$ composed of bold lines and $\bullet$ represent $f=x_1x_2x_3+\overline{x_2}x_3$.

\subsection{Sensitivity of Boolean Functions}
In this subsection, we will give several definitions about sensitivity, which will be used later in our Boolean matching method.

\begin{definition}
The \textit{sensitivity} of $f$ on the word $x$, which is also called \textit{local sensitivity}, is the number of input positions, changing any bit in which also changes the output: $s(f,x) = |{i:f(x) \neq f(x^i)}|$.
\end{definition}

If  $f(x)\neq f(x^i)$, we say $f$ and input $x$ is sensitive on index $i$.
We can further define the \emph{sensitivity} of $f$ as $s(f) = \max\{s(f,x): x \in \{0,1\}^n\}$, the $\emph{0}$-\emph{sensitivity} of $f$ as $s^0(f) = \max\{s(f,x): x \in \{0,1\}^n, f(x) = 0\}$ and the $\emph{1}$-\emph{sensitivity} of $f$ as $s^1(f) = \max\{s(f,x): x \in \{0,1\}^n, f(x) = 1\}$.

By the above definition, obviously we can get that for any Boolean function $f$ on $n$ variables, $s(f)$ is not greater than $n$. 
Also, it is trivially observed that this upper bound is tight, i.e., there are functions with sensitivity $n$.

\begin{example}
Let $f=x_1x_2x_3$, a 3-variable AND function.
For a word $x=\emph{000}$, $f(x)$ will not change no matter any bit changes, so $s(f,\mathit{000}) = 0$.
Furthermore, $s(f) = \max\{s(f,x)\} = s(f, \emph{111}) = 3$.
We also have $s^0(f) = s(f, \emph{101}) = 1$ and $s^1(f) = s(f, \emph{111}) = 3$.
\end{example}
\begin{definition}
We can define $average\;sensitivity$ $\widehat{s(f)}$ as
\begin{equation}
    \widehat{s(f)} = \frac{1}{2^n}\sum_{x}s(f,x).
\nonumber
\end{equation}
\end{definition}
$\widehat{s(f)}$ can be regarded as comprehensive measures of the sensitivity of Boolean functions.

\begin{figure}[tbp]
\centering
\begin{subfigure}[b]{0.20\textwidth}
\centering
\includegraphics[width=0.9\linewidth]{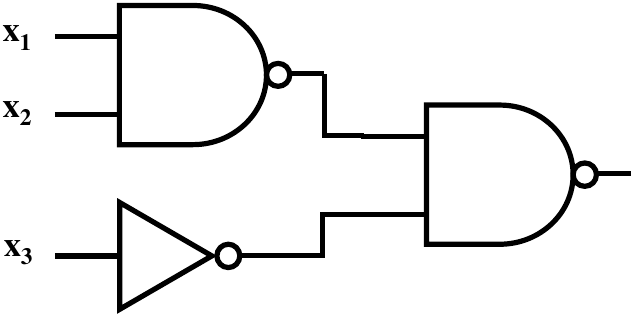}
\caption{$f = x_1x_2 + x_3$}
\label{fig:npneq1}
\end{subfigure}
\begin{subfigure}[b]{0.25\textwidth}
\centering
\includegraphics[width=0.9\linewidth]{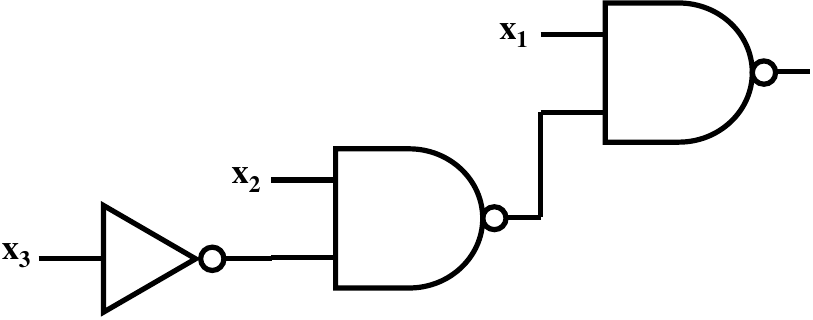}
\caption{$g = \overline{x_1} + x_2\overline{x_3}$}
\label{fig:npneq2}
\end{subfigure}
\caption{Two NPN-equivalent functions $f$ and $g$.}
\label{fig:npnexample}
\vspace{-1em}
\end{figure}

\subsection{Boolean Matching}

An NP transformation of a Boolean function is composed of variables negations and permutations.
Negation replaces a variable by its complement~(e.g., $x_1 \to \overline{x_1}$), which is also called \textit{flip}.
Permutation changes the order of variables~(e.g., $x_1x_2 \to x_2x_1$), which is also called \textit{swap}.
For an $n$-variable Boolean function, there are $2^n$ ways of transformations by flipping the inputs and $n!$ ways of transformations by swapping the variables.
Besides, there are two polarities of the function derived by complementing its output.
In total, there are $2^{n+1}n!$ transformations of the function by swapping its inputs and flipping its inputs and output.

\begin{definition}
Consider the set of all Boolean functions derived by the $2^{n+1}n!$ transformations of a Boolean function $f$, as described above. These functions constitute the \emph{NPN class} of function $f$. 
The \emph{NPN canonical form} of function $f$ is one function belonging to its NPN class, also called the representative of this class.
\end{definition}

The number of NPN classes is much smaller than the number of Boolean functions. 
For example, there are $2^{16}$ Boolean functions of 4 variables, and these functions can be split into 222 NPN classes.

\begin{definition}
Two Boolean functions $f$ and $g$ are \textit{NPN-equivalent}, $f \cong g$, if and only if there exists an NP transformation that satisfied $f(\pi((\neg)x_1,(\neg)x_2,\cdots,(\neg)x_n))=(\neg)g(x)$, where $\pi$ is a permutation and $(\neg)$ is an optional negation.
\end{definition} 
For simplicity, we denote $(\neg)x=(\neg)x_1(\neg)x_2\cdots(\neg)x_n$ in this paper.

\begin{example}
An example of NPN-equivalent functions is given in Fig~\ref{fig:npnexample}.
In this example, $f(x_1,x_2,x_3) = x_1x_2 + x_3$ and $g(x_1,x_2,x_3) = \overline{x_1} + x_2\overline{x_3}$ are
NPN-equivalent, because $f(\overline{x_3}, x_2, \overline{x_1})=g(x_1,x_2,x_3).$
\end{example}

If two Boolean functions are \emph{NPN equivalent}, one of them can be obtained from the other by swapping and flipping the inputs and the output.
The key task of Boolean matching is to determine whether two Boolean functions are \emph{NPN equivalent}.

\begin{figure*}[tbp]
\setlength{\belowcaptionskip}{0pt}
\centering
\begin{subfigure}[b]{0.19\textwidth}
\centering
\includegraphics[scale=0.33]{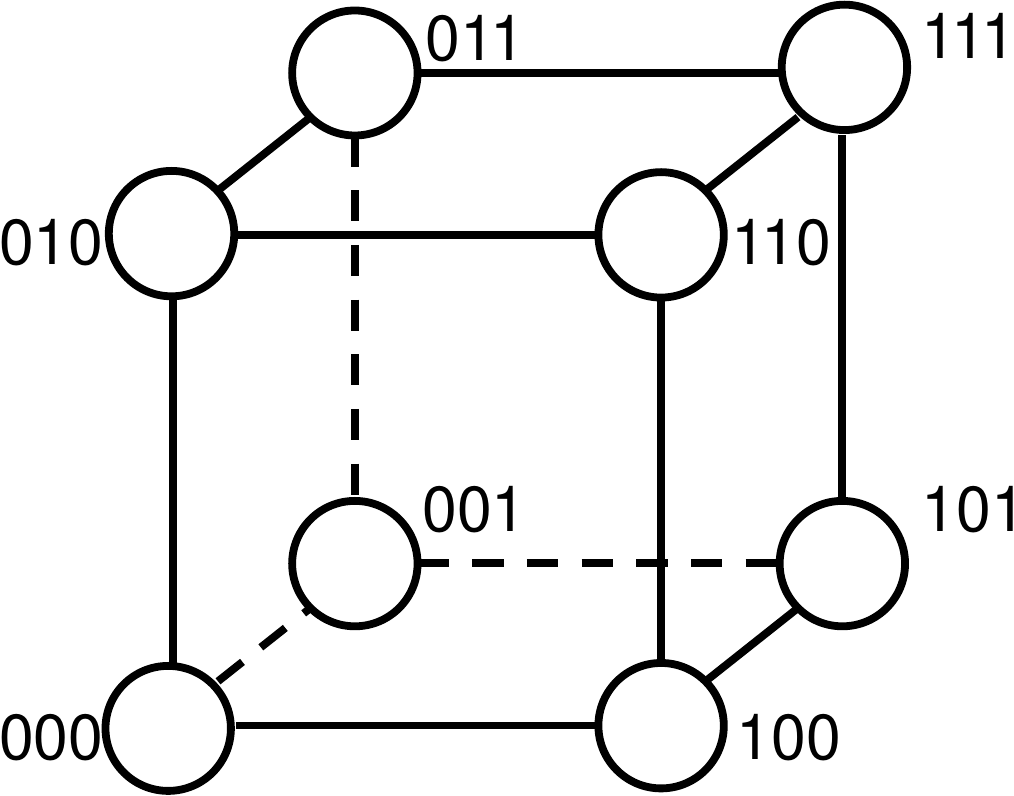}
\caption{\tiny{$s(f)=0, s^1(f)=0$\\$\widehat{s(f)}=0$\\$OSV(f)=\{\}$ \\ $OSV^1(f)=\{\}$}}
\label{fig:sensitivity1}
\end{subfigure}
\vspace{0.5em}
\begin{subfigure}[b]{0.19\textwidth}
\centering
\includegraphics[scale=0.33]{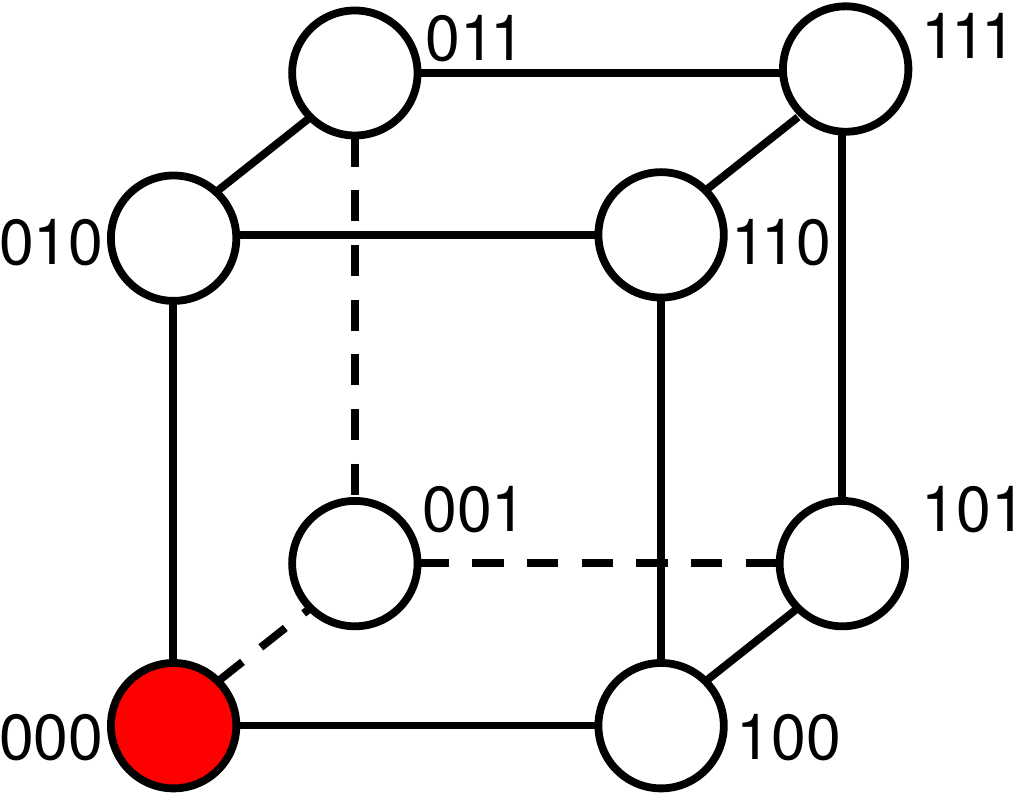}
\caption{\tiny{$s(f)=3, s^1(f)=3$\\$\widehat{s(f)}=0.75$\\$OSV(f)=\{3,1,1,1\}$ \\ $OSV^1(f)=\{3\}$}}
\label{fig:sensitivity2}
\end{subfigure}
\begin{subfigure}[b]{0.19\textwidth}
\centering
\includegraphics[scale=0.33]{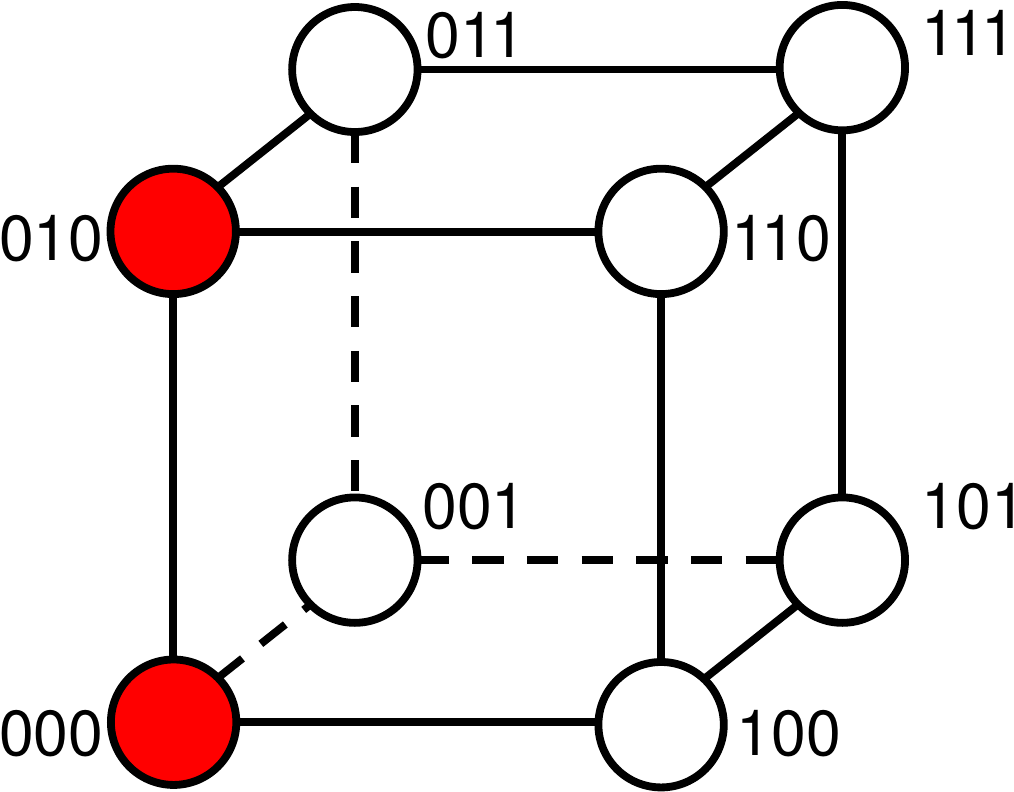}
\caption{\tiny{$s(f)=2, s^1(f)=2$\\$\widehat{s(f)}=1$\\$OSV(f)=\{2,2,1,1,1,1\}$ \\ $OSV^1(f)=\{2,2\}$}}
\label{fig:sensitivity3}
\end{subfigure}
\begin{subfigure}[b]{0.19\textwidth}
\centering
\includegraphics[scale=0.33]{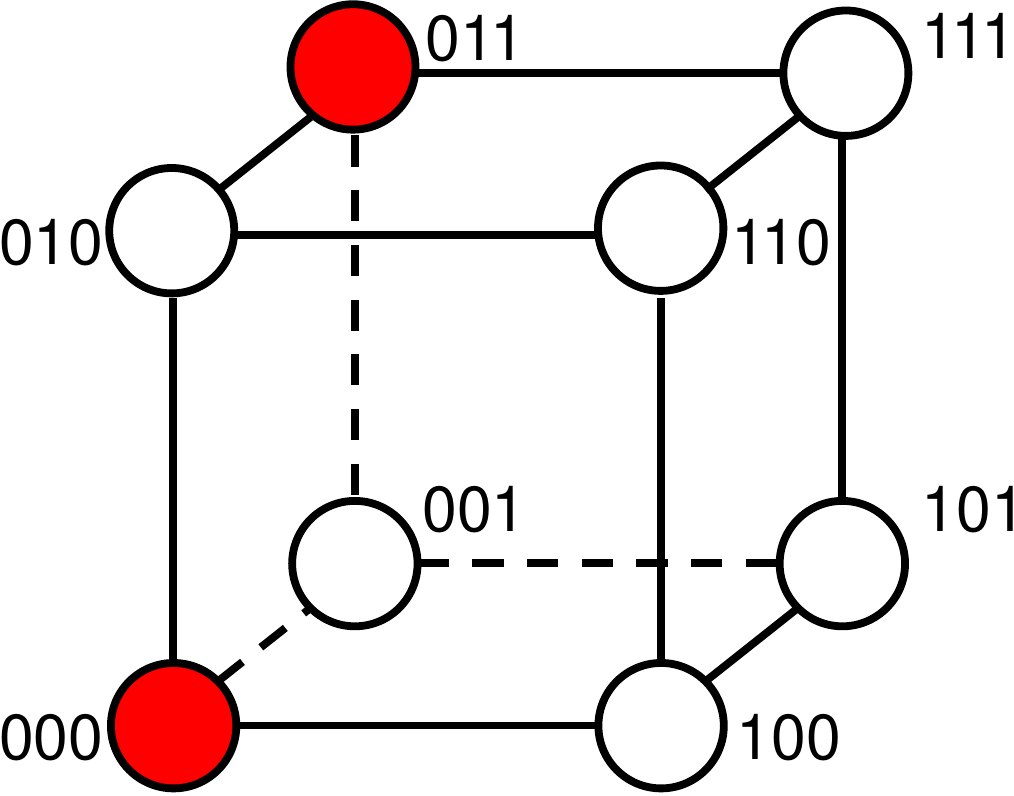}
\caption{\tiny{$s(f)=3, s^1(f)=3$\\$\widehat{s(f)}=1.5$\\$OSV(f)=\{3,3,2,2,1,1\}$ \\ $OSV^1(f)=\{3,3\}$}}
\label{fig:sensitivity4}
\end{subfigure}
\begin{subfigure}[b]{0.19\textwidth}
\centering
\includegraphics[scale=0.33]{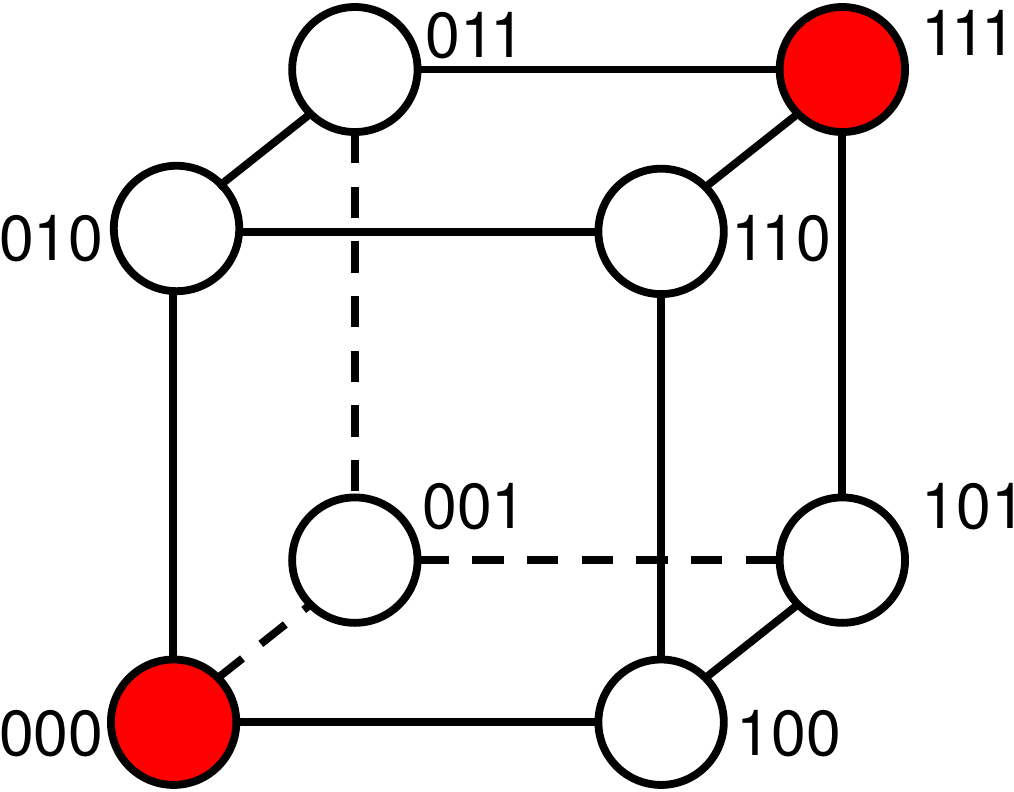}
\caption{\tiny{$s(f)=3, s^1(f)=3$\\$\widehat{s(f)}=1.5$\\$OSV(f)=\{3,3,1,1,1,1,1,1\}$ \\ $OSV^1(f)=\{3,3\}$}}
\label{fig:sensitivity5}
\end{subfigure}
\vspace{0.5em}
\begin{subfigure}[b]{0.19\textwidth}
\centering
\includegraphics[scale=0.33]{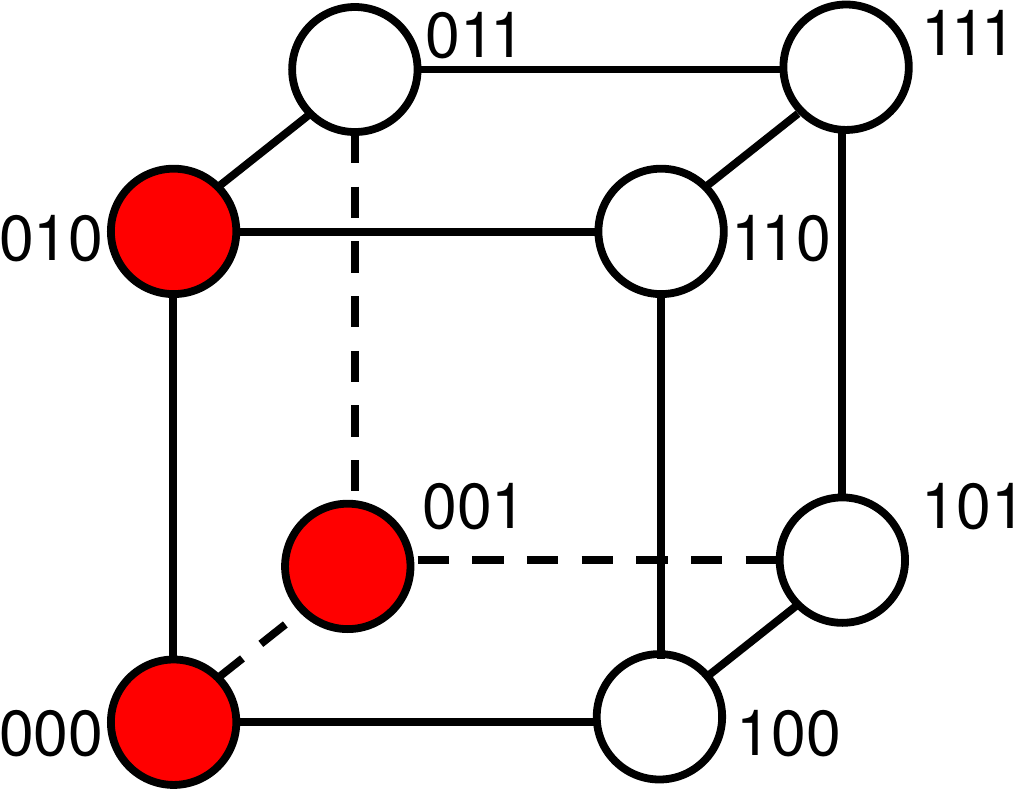}
\caption{\tiny{$s(f)=2, s^1(f)=2$\\$\widehat{s(f)}=1.125$\\$OSV(f)=\{2,2,2,1,1,1\}$ \\ $OSV^1(f)=\{2,2,1\}$}}
\label{fig:sensitivity6}
\end{subfigure}
\begin{subfigure}[b]{0.19\textwidth}
\centering
\includegraphics[scale=0.33]{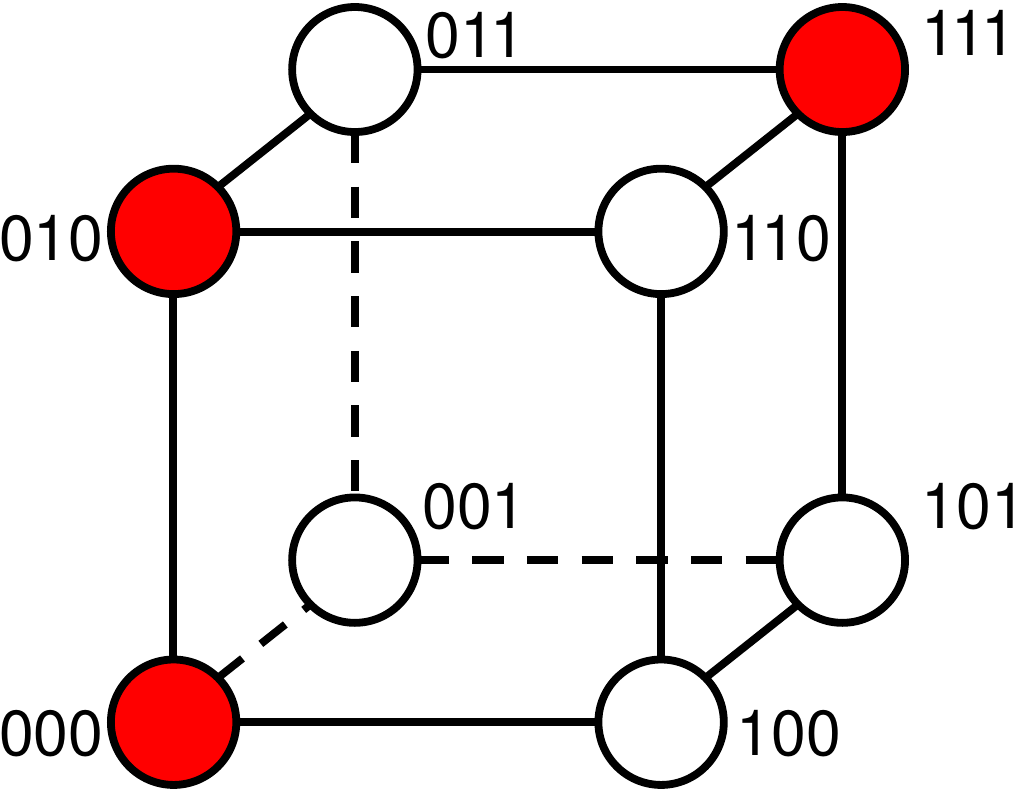}
\caption{\tiny{$s(f)=3, s^1(f)=3$\\$\widehat{s(f)}=1.75$\\$OSV(f)=\{3,2,2,\linebreak[0]2,2,\linebreak[0]1,1,1\}$ \\ $OSV^1(f)=\{3,2,2\}$}}
\label{fig:sensitivity7}
\end{subfigure}
\begin{subfigure}[b]{0.19\textwidth}
\centering
\includegraphics[scale=0.33]{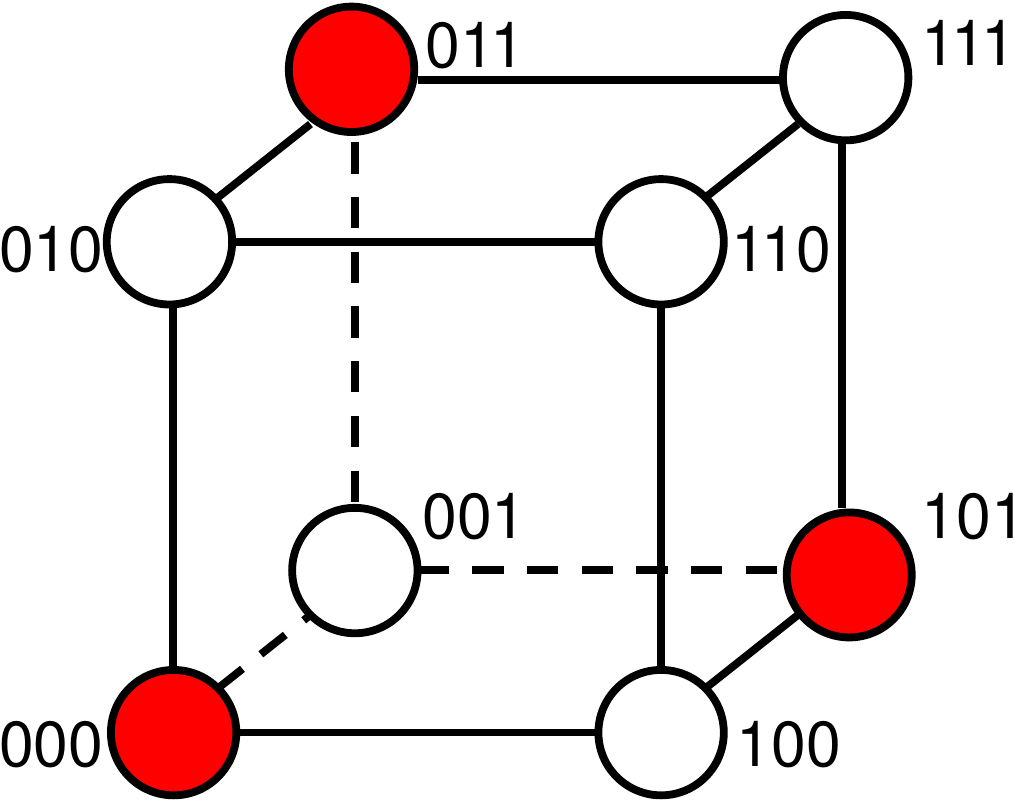}
\caption{\tiny{$s(f)=3, s^1(f)=3$\\$\widehat{s(f)}=2.25$\\ $OSV(f)=\{3,3,3,3,2,\linebreak[0]2,2\}$ \\ $OSV^1(f)=\{3,3,3\}$}}
\label{fig:sensitivity8}
\end{subfigure}
\begin{subfigure}[b]{0.19\textwidth}
\centering
\includegraphics[scale=0.33]{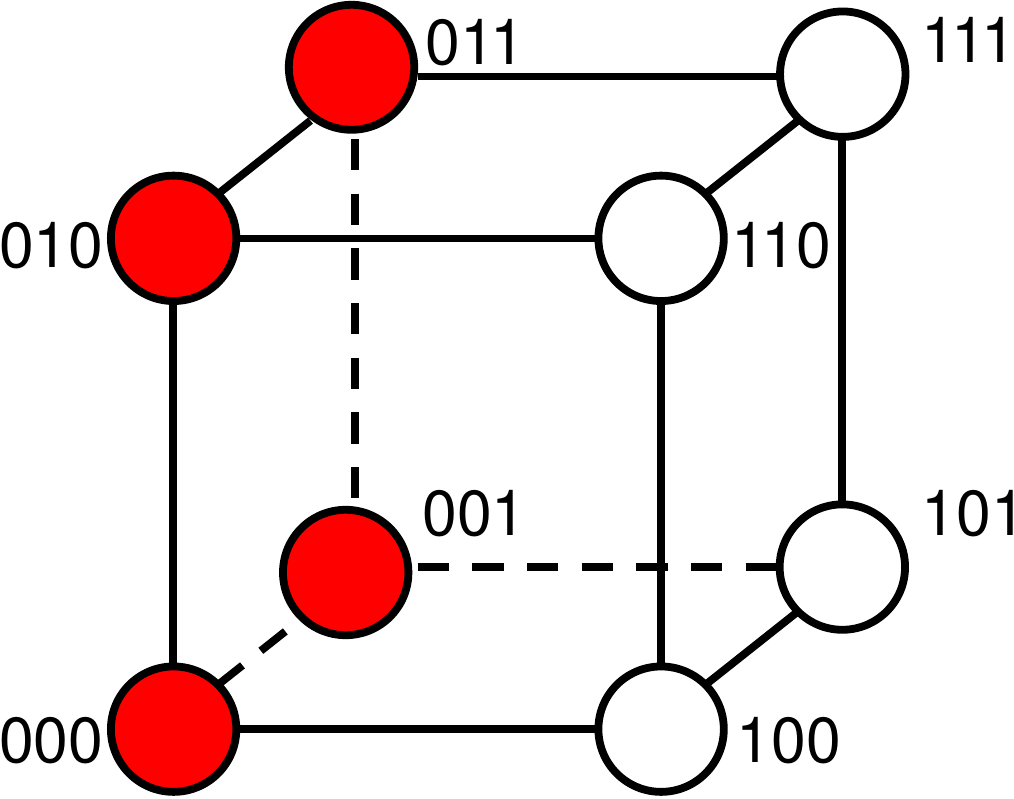}
\caption{\tiny{$s(f)=1, s^1(f)=1$\\$\widehat{s(f)}=1$\\$OSV(f)=\{1,1,1,1,1,1,\linebreak[0]1,1\}$ \\ $OSV^1(f)=\{1,1,1,1\}$}}
\label{fig:sensitivity9}
\end{subfigure}
\begin{subfigure}[b]{0.19\textwidth}
\centering
\includegraphics[scale=0.33]{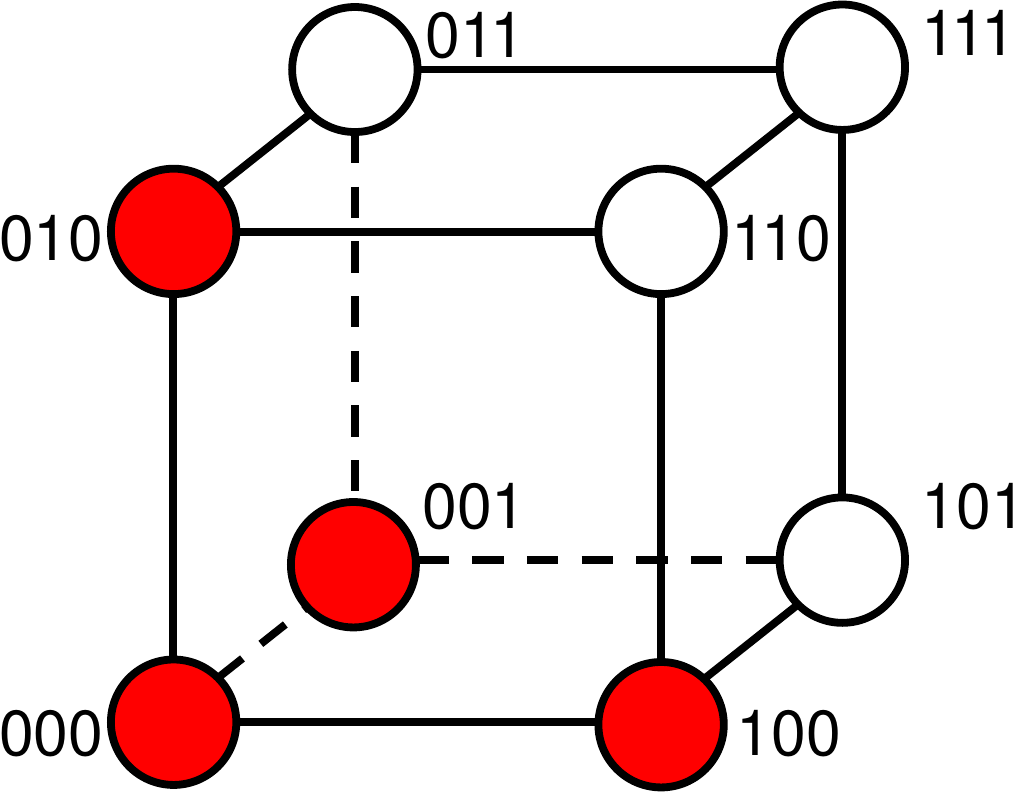}
\caption{\tiny{$s(f)=2, s^1(f)=2$\\$\widehat{s(f)}=1.5$\\$OSV(f)=\{2,2,2,2,2,2\}$\\$OSV^1(f)=\{2,2,2\}$}}
\label{fig:sensitivity10}
\end{subfigure}
\begin{subfigure}[b]{0.19\textwidth}
\centering
\includegraphics[scale=0.33]{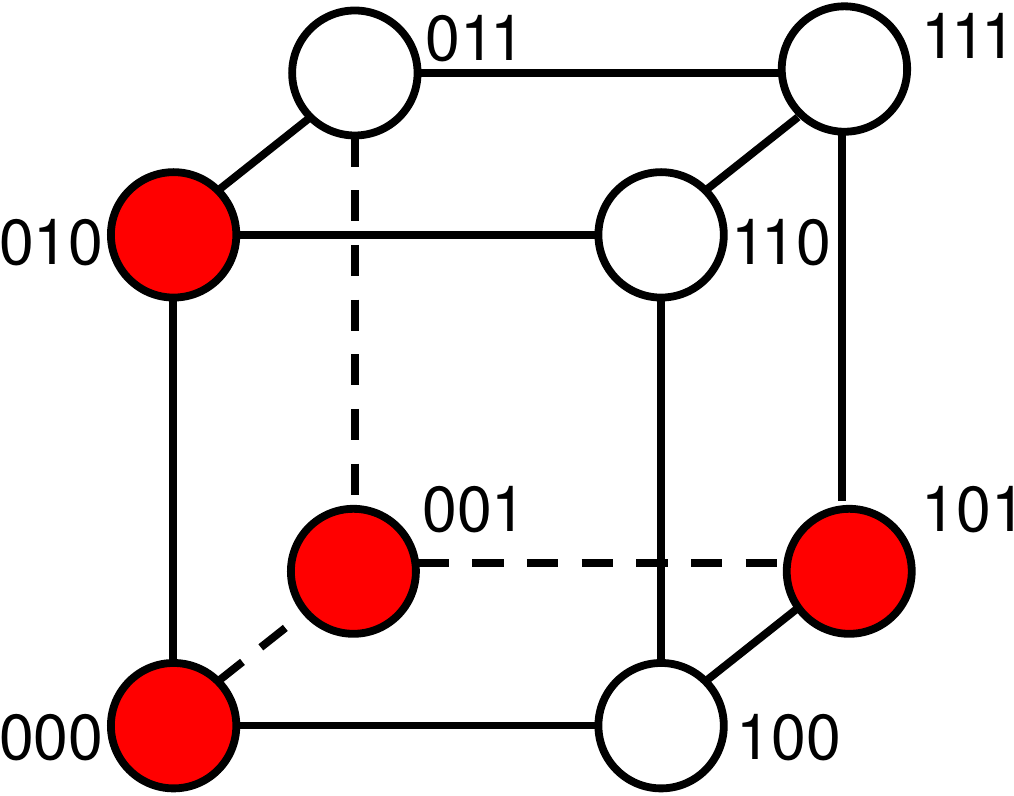}
\caption{\tiny{$s(f)=2, s^1(f)=2$\\$\widehat{s(f)}=1.5$\\$OSV(f)=\{2,2,2,2,1,1,\linebreak[0]1,1\}$\\$OSV^1(f)=\{2,2,1,1\}$}}
\label{fig:sensitivity11}
\end{subfigure}
\begin{subfigure}[b]{0.19\textwidth}
\centering
\includegraphics[scale=0.33]{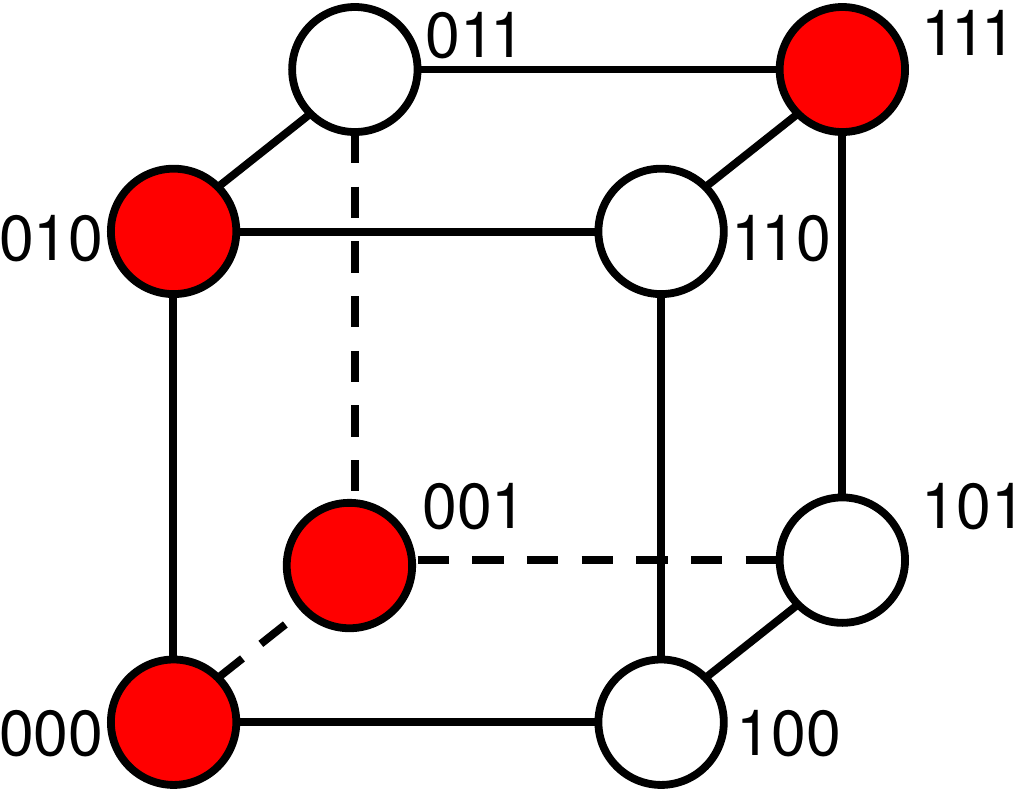}
\caption{\tiny{$s(f)=3, s^1(f)=3$\\$\widehat{s(f)}=2$\\$OSV(f)=\{3,3,2,2,2,2,\linebreak[0]1,1\}$ \\ $OSV^1(f)=\{3,2,2,1\}$}}
\label{fig:sensitivity12}
\end{subfigure}
\begin{subfigure}[b]{0.19\textwidth}
\centering
\includegraphics[scale=0.33]{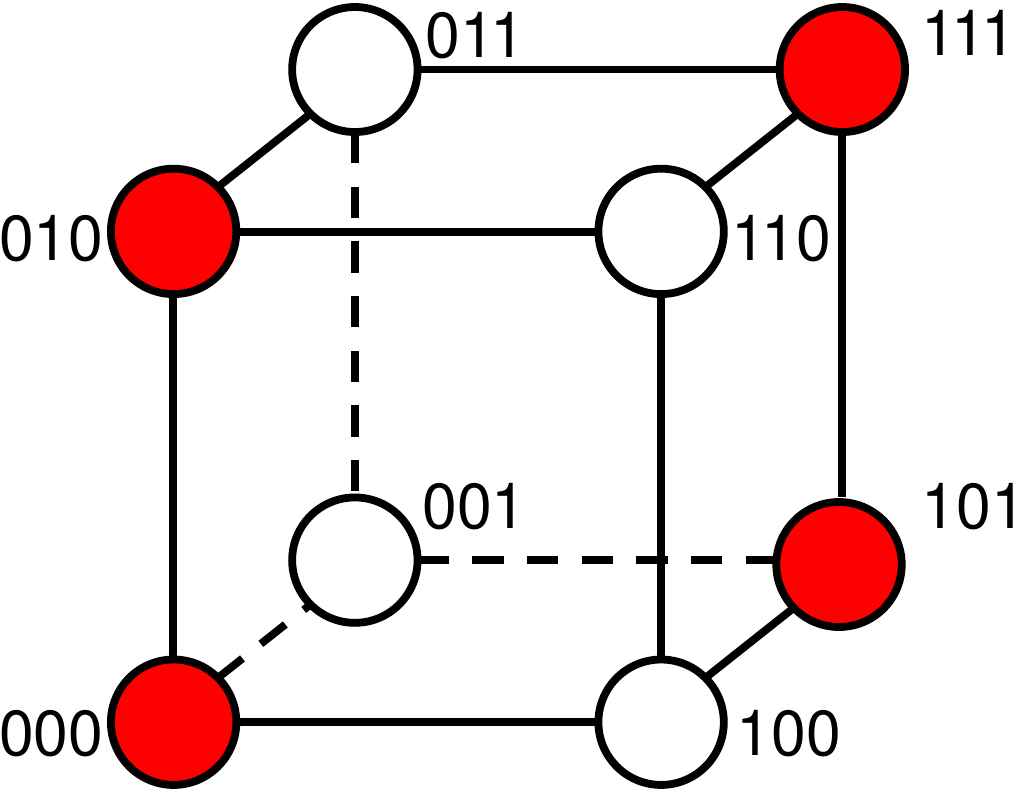}
\caption{\tiny{$s(f)=2, s^1(f)=2$\\$\widehat{s(f)}=2$\\$OSV(f)=\{2,2,2,2,2,2,\linebreak[0]2,2\}$ \\ $OSV^1(f)=\{2,2,2,2\}$}}
\label{fig:sensitivity13}
\end{subfigure}
\begin{subfigure}[b]{0.19\textwidth}
\centering
\includegraphics[scale=0.33]{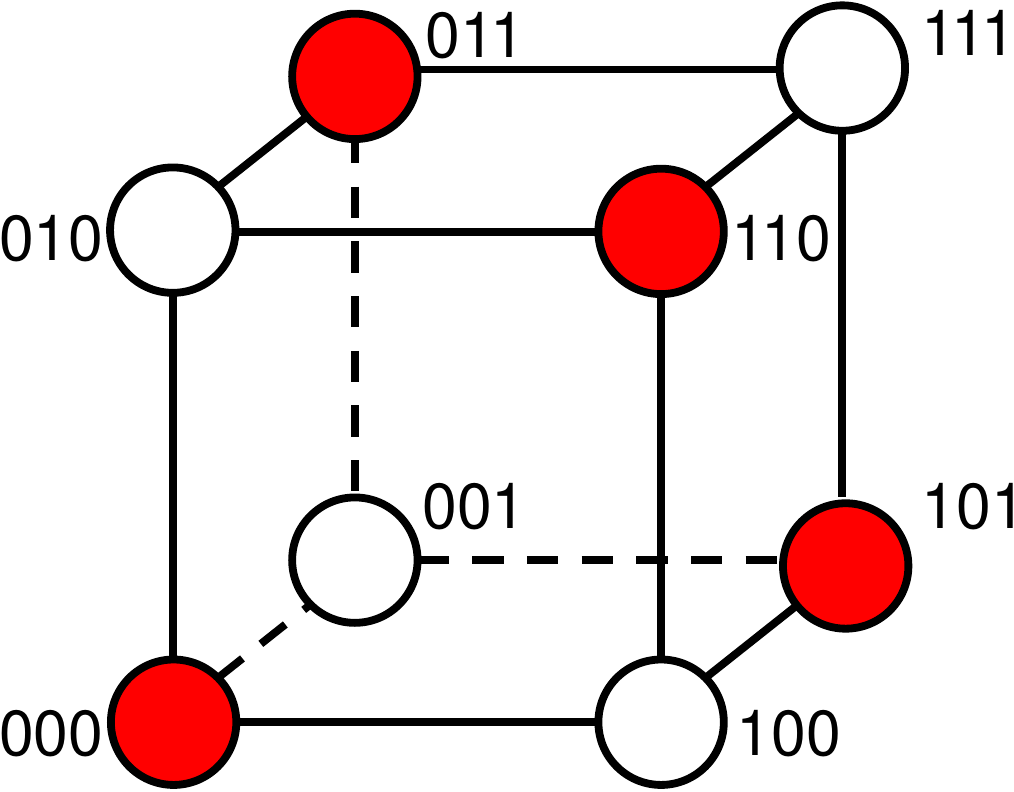}
\caption{\tiny{$s(f)=3, s^1(f)=3$\\$\widehat{s(f)}=3$\\$OSV(f)=\{3,3,3,3,3,3,\linebreak[0]3,3\}$ \\ $OSV^1(f)=\{3,3,3,3\}$}}
\label{fig:sensitivity14}
\end{subfigure}
\caption{256 3-input functions fall into 14 different NPN classes. This figure shows $s(f)$, $s^1(f)$, $OSV$ and $OSV^1$ of these 14 NPN equivalent classes . We omit \textit{0} in $OSV$ and $OSV^1$.From the subfigures, we can see that $OSV(f)$ of different NPN classes are totally different. Fig.~\ref{fig:sensitivity4} and Fig.~\ref{fig:sensitivity5}'s $OSV^1$ are the same, but their $AH(^3\!SD^1)$ are different.}
\label{fig:sensitivity_view}
\end{figure*}

\section{Sensitivity Properties}\label{sec:prove}

In this section, we provide some definitions of sensitivity-related signatures, and some theorems and their proofs about these signatures.
These theorems are the basis of the sensitivity-based pruning algorithm.
Because of the definition of sensitivity, the polarity of the output~(output negation transformation) can not be considered in the sensitivity, we can only take PN-equivalent into consideration.

\subsection{Basic Sensitivity Signatures}

\begin{lemma}
\label{Lm1} If Boolean function $f$ is  PN-equivalent  to Boolean function $g$, that is  $f(\pi((\neg)x_1,(\neg)x_2,\cdots,(\neg)x_n))=g(x)$,
then for any input $x$,   we have $$s(f,\pi((\neg)x))=s(g,x).$$ 
\end{lemma}

\begin{proof}
Since $f(\pi((\neg)x_1,(\neg)x_2,\cdots,(\neg)x_n))=g(x_1,x_2,\cdots, x_n)$, it is clear that 
if  $f$ and input $\pi((\neg)x)$ is sensitive on index $i$,
then $g$ and input $x$ will be sensitive on index $j$ such that $\pi(j)=i$. It is obvious to see that negation of a bit of an input can not change anything of a Boolean function's sensitivity. 

For example, let  $f(x)$ be a 4-bit Boolean function, permutation $\pi(1,2,3,4)=(4,3,2,1)$ and $f(\pi(\overline{x_1}x_2\overline{x_3}x_4)=g(x_1x_2x_3x_4)$.  Assume that $f$ and input $\pi(\overline{x_1}x_2\overline{x_3}x_4)=x_4\overline{x_3}x_2\overline{x_1}$ is sensitive on index $2$, we have $f(\pi(\overline{x_1}x_2\overline{x_3}x_4))=f(x_4\overline{x_3}x_2\overline{x_1})=g(x_1x_2x_3x_4)$ and $ \neg g(x_1x_2x_3x_4)=\neg f(x_4\overline{x_3}x_2\overline{x_1})=f(x_4x_3x_2\overline{x_1})=g(x_1x_2\overline{x_3}x_4)$. Therefore, Boolean function $g$ and input $x$ is sensitive on index $3=\pi(2)$.  

Therefore, for any $x$, it is clear that 
 $s(f,\pi((\neg)x))=s(g,x).$  
\end{proof}

\begin{theorem}
\label{Tm1}
Two PN-equivalent functions $f$ and $g$ have the same sensitivity, $0$-sensitivity and $1$-sensitivity: if $f$ is PN-equivalent to $g$, then $s(f) = s(g), s^0(f) = s^0(g)$ and $s^1(f) = s^1(g)$.
The contrapositive of this theorem is: if $s(f) \neq s(g), s^0(f) \neq s^0(g)$ or $s^1(f) \neq s^1(g)$, then $f$ is not  PN-equivalent to $g$. 
\end{theorem}

\begin{proof}
According to Lemma \ref{Lm1},  we have 
$s(g) = \max\{s(g,x): x \in \{0,1\}^n\}
=\max\{s(f,\pi((\neg)x)): x \in \{0,1\}^n\}
=\max\{s(f,x): x \in \{0,1\}^n\}=s(f)$.


Similarly, we can prove that $s^0(f) = s^0(g)$ and $s^1(f) = s^1(g).$
\end{proof}

\begin{definition}
For all words $X$ in truth table $T(f)$, we denote $OSV(f)=\Big(s(f,x^{(1)}),...,s(f,x^{(N)})\Big)$ such that $s(f,x^{(1)}) \ge \cdots \ge s(f,x^{(N)})$ as the \emph{ordered sensitivity vector}\footnote{Actually, it is a \emph{multiset}. But in order to describe it more intuitively, we call it a vector.} of function $f$, where $N=|X|$ is the total number of words.
\end{definition}

\begin{example}
For a 3-input Boolean function $f$, if we have $s(f,000) = s(f,101) = 3$, $s(f,001) = s(f,011) = s(f,100) = 2$, $s(f,010) = s(f,111) = 1$, and $s(f,110) = 0$, then $OSV(f) = \{ s(f,000), s(f,101), s(f,001), s(f,011) , s(f,100), s(f,010),\\ s(f,111), s(f,110) \} = \{3, 3, 2, 2, 2, 1, 1, 0\}$.
\end{example}

Similarly, we can define $OSV^0(f)$ as \textit{ordered} $\emph{0}$-\emph{sensitivity vector} and $OSV^1(f)$ as \textit{ordered} $\emph{1}$-\emph{sensitivity vector}.

\begin{theorem}\label{Tm2}
 Two PN-equivalent functions $f$ and $g$ have the same ordered sensitivity vector, ordered 0-sensitivity vector and ordered 1-sensitivity vector: if $f$ is PN-equivalent to  $g$, then $OSV(f) = OSV(g), OSV^0(f) = OSV^0(g)$ and $OSV^1(f) = OSV^1(g)$.
The contrapositive of this theorem is: if $OSV(f) \neq OSV(g), OSV^0(f) \neq OSV^0(g)$ or $OSV^1(f) \neq OSV^1(g)$, then $f \ncong g$.
\end{theorem}
\begin{proof}
Since  $f$ is PN-equivalent to $g$, according to Lemma \ref{Lm1}, there exist a permutation $\pi$,  for any input $x$,  such that $s(f,\pi((\neg)x))=s(g,x)$. For $\{x^1,x^2,\cdots,x^N\}=\{00\ldots0,00\ldots1,\cdots,11\ldots1\}=\{0,1\}^n$, let $y^i=\pi((\neg)x^i)$, it is obvious that $\{y^1,y^2,\cdots,y^N\}=\{00\ldots0,00\ldots1,\cdots,11\ldots1\}=\{0,1\}^n$. It is clear that the multiset $\{s(f,\pi((\neg) x^1)),s(f,\pi((\neg) x^2)),\cdots,s(f,\pi((\neg)x^N))\}=\{s(g,x^1),s(g,x^2),\cdots,s(g,x^N)\}=\{s(f,y^1),s(f,y^2),\ldots,\linebreak[0] s(f,y^N))\}.$  Therefore,   $OSV(f) = OSV(g)$. Similarly, we can have  $OSV^0(f) = OSV^0(g)$ and $OSV^1(f) = OSV^1(g)$.
\end{proof}
According to the proof of Theorem \ref{Tm2},  we have the following Corollary:

\begin{corollary}
\label{Co1}
Two PN-equivalent functions $f$ and $g$ have the same average sensitivity: if $f$ is PN-equivalent to $g$, then $\widehat{s(f)} = \widehat{s(g)}$.
\end{corollary}


\subsection{Advanced Sensitivity Signatures}

In the previous subsection, we only proved that $s(f)$ and $OSV(f)$ are prerequisites for NPN equivalence.
Therefore, to further distinguish the unmatched Boolean functions, we also design several advanced signatures based on $s(f)$ and $OSV(f)$.

\begin{definition}
A $K$-\emph{sensitivity domain} of $f$ $^K\!SD(f)$ contains all the words from truth table $T(f)$ that satisfied $s(f,x) = K$: $^K\!SD(f) = \{x|s(f,x) = K, x \in X\}$.
Similarly, we can define $K$-$\emph{0}$-\emph{sensitivity domain} and $K$-$\emph{1}$-\emph{sensitivity domain}
as $^K\!SD^0(f)$ and $^K\!SD^1(f)$, respectively.
\end{definition}

\begin{definition}
Let $Q_n$ be the $n$-dimensional hypercube graph.
We can get an induced subgraph $^K\!SG(f)$ from $Q_n$, whose vertices are words $x$ that satisfied $s(f,x) = K$.
We call $^K\!SG(f)$ as $K$-\emph{sensitivity graph} of $f$.
Similarly, we can define $K$-$\emph{0}$-\emph{sensitivity graph} and $K$-$\emph{1}$-\emph{sensitivity graph}
as $^K\!SG^0(f)$ and $^K\!SG^1(f)$ respectively.
\end{definition}

\begin{theorem}
\label{Tm3}
 If $f$ is PN-equivalent to $g$, then for any $K$, $^K\!SG(f)$ and $^K\!SG(g)$ are isomorphic.
Similarly, $^K\!SG^0(f)$ and $^K\!SG^0(g)$, $^K\!SG^1(f)$ and $^K\!SG^1(g)$ are isomorphic.
\end{theorem}
\begin{proof}
Let us recall the definition of isomorphism first.
An isomorphism of graphs $G$ and $H$ is a bijection between the vertex sets of $G$ and $H$
$$Bi:V(G)\to V(H)$$
such that any two vertices $u$ and $v$ of $G$ are adjacent in $G$ if and only if  $Bi(u)$ and $Bi(v)$ are adjacent in $H$.

Since $f$ is  PN-equivalent to $g$ and according Lemma \ref{Lm1},  there exists a permutation $\pi$, for any input $x$,  such that $s(f,\pi((\neg)x)=s(g,x)$.  

Suppose that $^K\!SD(f) = \{x|s(f,x) = K, x \in \{0,1\}^n\}$ and let $l={^K\!SD}(f) = \{x|s(f,x) = K, x \in \{0,1\}^n\}$. Assume that $^K\!SD(g)=\{x^{(1)},x^{(2)},\cdots,x^{(l)}\}$, we can get that 
$^K\!SG(f)=\{\pi((\neg)x^{(1)}), \pi((\neg)x^{(2)}),\cdots, \pi((\neg)x^{(l)})\}$ and $|^K\!SD(g)|=|^K\!SD(f)|$. It is clear that $x^{(i)}$ and $x^{(j)}$ have an edge only and only if $\pi((\neg)x^{(i)})$ and $\pi((\neg)x^{(j)})$ have an edge.  Therefore,  $^K\!SG(f)$  is isomorphic to $^K\!SG(g)$. 
\end{proof}

However, graph isomorphism has no polynomial-time exact algorithm yet, and we have to give some approximate methods to prove that the two graphs are not isomorphic.

\begin{definition}
We denote $|E(^K\!SG^1(f))|$ as the number of edges in $^K\!SG^1(f)$.
\end{definition}

According to the proof of Theorem \ref{Tm3}, we have the following Corollary:

\begin{corollary}
\label{Co2}
If $f$ is PN-equivalent to $g$, then $|E(^K\!SG^1(f))|$ = $|E(^K\!SG^1(g))|$.
\end{corollary}

\begin{definition}
\emph{Hamming distance}
$h(x,y)$ is a metric for comparing two binary strings $x$ and $y$. It is the number of bit positions in which $x$ and $y$ are different.
We define $AH(^K\!SD^1(f))$ as the \emph{average Hamming distance} of a K-sensitivity domain:
\begin{equation}
    AH(^K\!SD^1(f)) = \frac{1}{N}\sum_{x,y\in ^K\!SD^1(f)}h(x,y). 
\nonumber
\end{equation}

\end{definition}

\begin{corollary}
\label{Co3}
If $f$ is PN-equivalent to $g$, then $AH(^K\!SD^1(f))$ = $AH(^K\!SD^1(g))$.
\end{corollary}

\begin{proof}
It is easy to see that $h(x,y)=h(\pi((\neg)x,(\neg)y))$. 
According to Lemma \ref{Lm1} and the definition,  therefore the corollary holds. 
\end{proof}

Overall, we can determine in advance that two Boolean functions are not NPN-equivalent in Boolean matching through \textbf{Theorem~\ref{Tm1}} (sensitivities $s$), \textbf{Theorem~\ref{Tm2}} (ordered sensitivity vectors $OSV$), \textbf{Theorem ~\ref{Tm3}} (isomorphism of $K$-sensitivitiy graphs $^K\!SG$), \textbf{Corollary~\ref{Co1}} (average sensitivities $\widehat{s}$), \textbf{Corollary~\ref{Co2}} (edge counts of $K$-sensitivity graphs $|E(^K\!SG^1)|$), and \textbf{Corollary~\ref{Co3}} (average Hamming distances $AH(^K\!SD^1)$).


Fig.~\ref{fig:sensitivity_view} shows the results of several sensitivity-based signatures of 3-input Boolean functions.
3-input Boolean functions fall into 14 different NPN classes.
From this figure, we can see that $OSV(f)$ of different NPN classes are totally different. 
Fig.~\ref{fig:sensitivity4} and Fig.~\ref{fig:sensitivity5}'s $OSV^1$ are the same, but their $AH(^K\!SD^1)$ are different.
That is to say, we can completely distinguish all NPN classes by signature OSV.

\subsection{Symmetry, Cofactor Signatures vs. Sensitivity Signatures}
\label{sec:compare}

Let $f$ be an $n$-bit Boolean function: $\{0,1\}^n\to \{0,1\}$, two variables $x_i$ and $x_j$ are symmetric if and only if $f(...,x_i,...,x_j,...) = f(...,x_j,...,x_i,...)$ \cite{huang2013fast}.  Let $g$ be an $n$-bit Boolean function: $\{0,1\}^n\to \{0,1\}$ and $g$ is NPN-equivalent to $f$. Without loss of generality, assuming that  $f(\pi(x))=g(x)$, if $x_i$ and $x_i$ are symmetric in function $g$, then we have $g(...,x_i,...,x_j,...) = g(...,x_j,...,x_i,...)$.  Since  $f(\pi(x))=g(x)$, we can get that $f(...,\pi(x_i),...,\pi(x_j),...)=f(...,\pi(x_j),...,\pi(x_i),...)$. Therefore,  $\pi(x_i)$ and $\pi(x_j)$ are symmetric in function $f$. It is clear that symmetric group structure including the number of groups and the size of each group are the same if two functions $f$ and $g$ are NPN-equivalent. 

Zhang et al.~\cite{zhang2019efficient} considered structural cofactor signature of Boolean functions. In their paper, they defined a cofactor as $f_i(f_{\overline{x_i}})$, which can be seen as a face of the hypercube $Q_n$ that represent $f$. 
Many Boolean matching works~\cite{huang2013fast, zhang2019efficient, chai2006building, abdollahi2008symmetry} focused on face properties of the hypercube $Q_n$.  We investigate sensitivity of Boolean functions in this paper, which is the point structure of the hypercube $Q_n$, and focus on connections between points of value 0 and value 1. The method used symmetry and this paper are mutually complementary. 

Since  hypercube $Q_n$ has $2n$ faces and $2^n$ points, 
instinctively, there are $2n$ items of information when one uses structural cofactor signatures and symmetric and there are $2^n$ items of information when one uses sensitivity signatures. Therefore,   sensitivity signatures are expected to be more efficient.  However, the time complexity to compute sensitivity signatures is not more than to compute cofactor signatures of a Boolean function. They both need $O(2^n)$.  

\section{Methodology}\label{sec:method}
This section shows how to use sensitivity-based signatures described in the previous section to reduce the search space as much as possible, which speeds up NPN equivalence checking.

\begin{algorithm}[tbp]
\caption{Fast Sensitivity Computation}
\label{algorithm:calculation}
\begin{algorithmic}[1]
\Require Truth table $T(f)$ of an $n$-variable Boolean function $f$
\Ensure $s(f)$, $\widehat{s(f)}$, $OSV(f)$

\Comment{\emph{with a compression factor of $z=32$ or $64$}}
\State $len \gets max(2^n / z, 1)$
\State $T_c[0:len] \gets \mathit{compress}(T(f))$
\State $\mathit{sum\_sensitivity} \gets 0$
\State $d \gets -z$
\For{$i = 0$ to $len$}
    \State $v \gets T_c[i]$
    \State $d \gets d + z$
    \While{$v \ne 0$}
        \State $ss \gets 0$
        \State $flip\_lowest\_one(v)$
        \For{$index$ in $n$}
            \State $flip(T_{index})$
            \If{$check\_output\_flip()$}
                \State $ss \gets ss + 1$
            \EndIf
        \EndFor
        \State $s(f) \gets \max(s(f),\;ss)$
        \State Update $OSV(f)$, $\sum(s(f))$, $ss$ and $d$
    \EndWhile
\EndFor
\State $\widehat{s(f)} = getAve(\sum{s(f)})$
\State $OSV(f) = OSV.sort()$
\State \Return $s(f)$, $\widehat{s(f)}$, $OSV(f)$
\end{algorithmic}
\end{algorithm}

\subsection{Fast Sensitivity Computation}
As the property of Boolean sensitivity as mentioned above, we find that it is very convenient for us to implement the code based on binary string. 
The length of the truth table we defined as \textit{len}, the inputs variables number is as the defined \textit{n} and it meet the equation of $ len = 2^n $. 
So if we want to do the three types of transformations (negate inputs, permute inputs and negate outputs), it could be completed in $O(len \cdot k \cdot n)$ time through the bit operation on string and k is the number of the flipping position in an unsigned integer. 
Therefore, we could also perform some string-related optimization, like bits compress, to reduce the processing time of sensitivity computation.
If we use BDD to represent a Boolean function, there is no such advantage.

\emph{Algorithm~\ref{algorithm:calculation}} presents an efficient procedure to compute $s(f)$.
We can compute $s^0(f)$ and $s^1(f)$ similarly.
The algorithm takes truth table $T(f)$ of a Boolean function as input.
First, it compresses the truth table with a compression factor $z$~(usually we set $z$ to 32 or 64) and initializes the sensitivity as well as the counter~(Line 1-4).
Then, for each item in compressed truth table $T_c$, the procedure flip the item~(Line 10-12), check the output~(Line 13), and get the temporary sensitivity~(Line 14-17).
Next it updates $OSV(f)$, the sum of sensitivity $\sum{f(x)}$, temporary sensitivity and counter~(Line 18).
At last, we get $\widehat{s(f)}$ by $\sum{s(f)}$ and ordered $OSV(f)$.

\begin{algorithm}[tbp]
\caption{Basic Sensitivity Signatures Pruning}
\label{algorithm:basicpruning}
\begin{algorithmic}[1]
\Require Truth tables $T(\cdot)$ of $n$-variable Boolean functions $f$ and $g$
\Ensure  False (when $f \ncong g$) or Unknown

\State Compute $s(\cdot)$, $\widehat{s}(\cdot)$, and $OSV(\cdot)$ of $f$ and $g$ using Algorithm~\ref{algorithm:calculation}
\If{$s(f) \neq s(g)$}
    \State \Return False
\ElsIf{$\widehat{s(f)} \neq \widehat{s(g)}$ or $OSV(f) \neq OSV(g)$}
    \State \Return False
\EndIf
\State \Return Unknown
\end{algorithmic}
\end{algorithm}

\begin{algorithm}[tbp]
\caption{Advanced Sensitivity Signatures Pruning}
\label{algorithm:advancedpruning}
\begin{algorithmic}[1]
\Require Two ordered sensitivity vector $OSV(f)$ and $OSV(g)$ of two Boolean functions $f$ and $g$, maximum iteration $maxIter$
\Ensure  False (when $f \ncong g$) or Unknown

\State Get maximum local sensitivity K
\While{$i < maxIter$}
　　\If{$|E(^K\!SG(f))| \neq |E(^K\!SG(g))|$}
        \State \Return False
    \ElsIf{$AH(^K\!SD(f)) \neq AH(^K\!SD(g))$}
        \State \Return False
    \EndIf
    \State Get next $K$
    \State $i++$
\EndWhile
\State \Return Unknown
\end{algorithmic}
\end{algorithm}

We will give an example.
Assume a 5-input Boolean function $g$ has a truth table ``11000100000101100011101100010110'', which has $2^5{=}32$ bits with $g(00000)$ at the leftmost bit.
This truth table can be implemented efficiently by packing multiple bits in an entry. Assume we pack every 8 bits in an entry, the truth table is compressed into a 4-entry array [``11000100", ``00010110", ``00111011", ``00010110"] = [196, 22, 59, 22]. Moreover, we can perform the transformations on the compressed entries more efficiently than on a single bit.
In practice, we can compress a truth table of $2^n$ bits into an uint32 array of length $2^n/32$ and attain an about $5\times$ speedup than the normal sensitivity calculation method.

\subsection{Sensitivity Pruning}
We use the sensitivity properties proved in Section~\ref{sec:prove} to derive the sensitivity signatures pruning.
Algorithm~\ref{algorithm:basicpruning} shows the pruning process based on basic sensitivity signatures.
The algorithm takes truth tables of two Boolean functions $f$ and $g$ as inputs.
The program first calculates sensitivity using Algorithm~\ref{algorithm:calculation} and compare the sensitivity of the two functions.
If $s(f) \neq s(g)$, then the procedure returns $f \ncong g$.
Otherwise, it gets average sensitivity and ordered sensitivity vector for comparison.
The procedure will return $f \ncong g$ if these two signatures are not equal.
If all these three signatures are equal, NPN equivalence will be tested by the follow-up signatures.
This algorithm is suitable for sensitivity, 0-sensitivity and 1-sensitivity.

In Figure~\ref{fig:sensitivity_view}, we can see that all 3-input NPN canonical forms could be constructed via $OSV$.
However, if we use $OSV^1$ as basic signatures, class~\ref{fig:sensitivity4} and class~\ref{fig:sensitivity5} could not be tested.
As said before, to further distinguish the unmatched Boolean functions, we also design advanced signatures based on ordered sensitivity vectors.

Algorithm~\ref{algorithm:advancedpruning} gives the advanced sensitivity signatures pruning method.
The program takes two ordered sensitivity vectors $OSV(f)$ and $OSV(g)$ of two Boolean functions $f$ and $g$ as well as the maximum iteration as inputs.
The maximum iteration is less than the number of elements with different values in the ordered sensitivity vector.
First, it gets the maximum local sensitivity $K$.
Obviously, it is the element at top of the vector.
Then we compare $|E(^K\!SD)|$ and $AH(^K\!SD)$ of the two Boolean functions one by one.
The procedure will return $f \ncong g$ if any of these two signatures are not equal.
Otherwise, the algorithm will get the next $K$~(the next small local sensitivity) and repeat Line 3-9 until maximum iteration reaches.

For example, if we only use 1-sensitivity to test the NPN equivalence of class~\ref{fig:sensitivity4} and class~\ref{fig:sensitivity5}, we can not get that class~\ref{fig:sensitivity4} and class~\ref{fig:sensitivity5} are not equivalent.
The $|E(^K\!SD)|$ of these two NPN classes are both equal to 0.
But we can know that these are two NPN classes because their $AH(^3\!SD)$ are not the same.

\subsection{Integration to Canonical Form Method}
The above pruning method can only quickly determine that two Boolean functions $f$ and $g$ belong to different NPN classes.
However, sensitivity properties are only prerequisites of NPN equivalence.
We can use these properties to efficiently determine the non-equivalence of Boolean functions but cannot get NPN-equivalent classes.
Therefore we adopt a fast canonical form-based method~\cite{huang2013fast} to complete the follow-up to test two Boolean functions that are NPN equivalent.
Please refer to this article~\cite{huang2013fast} for details due to space limitations.

\subsection{Overall Algorithm}
Algorithm~\ref{algorithm:overall} depicts our overall Boolean matching procedure, which is divided into four phases.
The first three phases are the pruning stages to test sensitivity signatures and reject non-NPN-equivalent functions.
The last phase verifies NPN equivalence using the canonical form.
In the pruning stages, once $f$ and $g$ fail any sensitivity signature test, the procedure returns false.

In the first phase, a well-known signature used in our matching procedure is the number of onset minterms.
Many literatures use this quantity as a first-order signature to determine the canonical form of Boolean functions~\cite{huang2013fast, chai2006building, abdollahi2008symmetry}.

For the output polarity assignment of a given function $f$, we consider both $|f|$ and $|\overline{f}|$. 
If $|f| < |\overline{f}|$, then we first apply the 1-\textit{sensitivity} remaining pruning algorithm to $|\overline{f}|$, and else we use 0-\textit{sensitivity}.
The reason why we first use 1-\textit{sensitivity} or 0-\textit{sensitivity} is that such a program can reduce the time to calculate the sensitivity, thereby speeding up the matching process.
If we can not test that two Boolean functions are not NPN equivalent only by 0-sensitivity and 1-sensitivity, we will apply sensitivity for further testing.

Assuming that we cannot test whether $f$ and $g$ are not equivalent after phase 2, we can apply phase 3 for further testing.
However, this phase is time-consuming, so we set it optional.
At last, we will apply a traditional symmetry-based canonical form method to make sure that $f$ and $g$ are NPN equivalent.

\begin{algorithm}[tb]
\caption{Overall Boolean Matching Algorithm}
\label{algorithm:overall}
\begin{algorithmic}[1]
\Require Truth tables $T(\cdot)$ of $n$-variable Boolean functions $f$ and $g$
\Ensure True (when $f \cong g$) or False (when $f \ncong g$)

\Comment{\emph{phase 1: prune by minterm signature}}
\If{$|f| \neq |g|$ and $|f| \neq |\overline{g}|$}
    \State \Return False
\EndIf

\Comment{\emph{phase 2: prune by basic sensitivity signatures}}
\If{$|f| < |\overline{f}|$}
    \State Prune as Algorithm~\ref{algorithm:basicpruning} using the 1-sensitivities
\Else
    \State Prune as Algorithm~\ref{algorithm:basicpruning} using the 0-sensitivities
\EndIf
\State Prune as Algorithm~\ref{algorithm:basicpruning} using the sensitivities

\Comment{\emph{phase 3 (optional): prune by advanced sensitivity signatures}}
\If{$|f| < |\overline{f}|$}
    \State Prune as Algorithm~\ref{algorithm:advancedpruning} using the 1-sensitivities
\Else
    \State Prune as Algorithm~\ref{algorithm:advancedpruning} using the 0-sensitivities
\EndIf
\State Prune as Algorithm~\ref{algorithm:advancedpruning} using the sensitivities

\Comment{\emph{phase 4: construct canonical form}}
\If{$\mathit{canonical}(f) = \mathit{canonical}(g)$}
    \State \Return True
\Else
    \State \Return False
\EndIf
\end{algorithmic}
\end{algorithm}

\section{Evaluation}\label{sec:evaluation}

\subsection{Environmental Setup}
We implement a sensitivity pruning algorithm in C++ and reimplement a fast symmetry-based fast Boolean matching method~\cite{huang2013fast} as the phase 4 in Algorithm~\ref{algorithm:overall}.
The whole procedure runs on an Intel Xeon 2-CPU 10-core computer with 60GB RAM. 
We generate Boolean functions of different bits to test the algorithm.
The truth tables of these Boolean functions are provided in a text file, one per line, which lists them one after another without separators.

\subsection{Boolean Function Generation}
We generated two groups of $n$-variable Boolean functions.
Considering the running time, the number of generated Boolean functions will gradually decrease when $n$ becomes larger. 
The first group is completely randomly generated, denoted as $Group\;1$.
However, there are a huge amount of NPN classes when $n$ increases and it is difficult for randomly generated examples to have NPN classes.
In practical applications, there will be a small number of NPN classes.
For example, the first step of technology mapping is to compute the canonical forms of the library cell functions in advance.
In the technology mapping step, the procedure will check the logic function of the subgraph in the subject graph is NPN equivalent to these canonical forms.
The number of library cell functions will not be large, so the number of NPN classes is also limited.

We generated another group of Boolean functions with about 100 NPN classes.
We directly use the nature of NPN equivalence and randomly adopt input flip, output flip and randomly input swap for all words of a truth table.
We randomly pick a certain number of functions from the $Group\;1$, and apply multiple NPN transformations for each function to get some NPN equivalent Boolean functions.
Then we can get another group of Boolean functions with NPN equivalent ones, denoted as $Group\;2$.

\subsection{Experimental Results}
We test the NPN matching procedure on both $Group\;1$ and $Group\;2$.
The compression factor is set to 32 and the maximum iteration is set to 3.
Table~\ref{tab:random} shows the effect of sensitivity signatures in reducing search space.
We adopt the concept of collision in hash.
We say that there is a collision if two Boolean functions $f$ and $g$ can not be determined to be mismatched after one pruning phase.
The columns “\#Coll. a. P2”, “\#Coll. a. P3” list the number of collisions after pruning phase 2 and phase 3 in Algorithm~\ref{algorithm:overall}.
The columns “\#Coll. a. Sym”, “\#Coll. a. H-Sym” list the number of collisions after pruning using symmetry and high-order symmetry in~\cite{huang2013fast}.
For each $n$, we select a certain amount of Boolean functions pairs from $Group\;1$, and apply Boolean matching.
Without loss of generality, we do Boolean matching 10 times and take the average.
The results show that sensitivity signatures could prune more mismatched Boolean functions than symmetry signatures.
Especially for large bits, only $OSV$ can prune most of the mismatched Boolean functions.

\begin{table}[tbhp]
\caption{Collisions of $n$-variables Boolean functions matching using $Group\;1$}
\label{tab:random}
\centering
\small
\begin{tabular}{IcIcIc|cIc|cI}
\shline
N   & \#Matching & \makecell[c]{\#Coll. \\a. P2}  & \makecell[c]{\#Coll. \\a. P3} & \makecell[c]{\#Coll. \\a. Sym} & \makecell[c]{\#Coll. \\a. H-Sym}\\
\shline
5   & 1M    & 1913  & 21    & 124063 & 70379\\
6   & 1M    & 98    & 0     & 62659  & 8049\\
7   & 1M    & 4     & 0     & 31307  & 2644\\
8   & 1M    & 0     & 0     & 15686  & 1589\\
\shline
9   & 100k  & 0     & 0     & 845    & 42\\
10  & 100k  & 0     & 0     & 382    & 22\\
11  & 100k  & 0     & 0     & 171    & 6\\
12  & 100k  & 0     & 0     & 115    & 1\\
\shline
13  & 10k   & 0     & 0     & 5      & 0\\
14  & 10k   & 0     & 0     & 3      & 1\\
15  & 10k   & 0     & 0     & 2      & 0\\
16  & 10k   & 0     & 0     & 3      & 0\\
\shline
\end{tabular}
\end{table}

Table~\ref{tab:npnmatch} gives the runtime of the proposed Boolean matching method in Algorithm~\ref{algorithm:overall} using $Group\;2$.
The column “$t_2$” is the runtime of Phase 1 and Phase 2 while the column “$t_3$” is the runtime of Phase 3.
Table~\ref{tab:random} shows that only $OSV$ can prune most of the mismatched Boolean functions, so we could omit this Phase 3 to reduce runtime.
The column “$t_{total}$ W/O P3” lists the runtime of Algorithm~\ref{algorithm:overall} without Phase 3~(Phase 1 + Phase 2 + Phase 4).
The column “$t_{total}$ W P3” lists the runtime of Algorithm~\ref{algorithm:overall} with Phase 3~(Phase 1 + Phase 2 + Phase 3 + Phase 4).
The column “$t_{base}$” is the runtime of~\cite{huang2013fast}, without integrating our sensitivity signatures pruning.
For each $n$, we also select a certain amount of Boolean functions pairs from $Group\;2$, and apply Boolean matching.
It is worth noting that we will ensure that about 15\% of the Boolean function pairs in these matches are NPN equivalent to test the effectiveness of the algorithm in practical applications.
Without loss of generality, we also do Boolean matching 10 times and take the average.
From Table~\ref{tab:npnmatch}, we can see that after integrating our sensitivity signatures pruning, the Boolean matching performs up to 3.85x speedup compared with previous work.

\begin{table*}[tbhp]
\caption{Runtime of $n$-variables Boolean functions matching using $Group\;2$}
\label{tab:npnmatch}
\centering
\small
\begin{tabular}{IcIcIc|cIc|cIcIc|cI}
\shline
N   & \#Matching & \makecell[c]{$t_2$\\(ms)}  & \makecell[c]{$t_3$\\(ms)} &\makecell[c]{$t_{total}$ W/O P3\\(ms)}& \makecell[c]{$t_{total}$ W P3 \\(ms)} & \makecell[c]{$t_{base}$\\(ms)}~\cite{huang2013fast} & \makecell[c]{Speedup \\W/O P3} & \makecell[c]{Speedup \\W P3}\\
\shline
5   & 1M    & 24.93 &	88.24	&  241.62  &	263.66  &	675.05	 & 2.79 & 2.56  \\
6   & 1M    & 25.15 &	15.23	&  132.76  &	150.51  &	303.54   & 2.28 & 2.02  \\
7   & 1M    & 25.53 &	34.64	&  305.16  &	377.29  &	1174.09  & 3.85 & 3.11  \\
8   & 1M    & 20.08 &	27.69	&  1130.19 &	1165.70 &	1257.48  & 1.11 & 1.08  \\
\shline
9   & 100k  & 3.73  &	3.06	&  320.64  &	322.80  &	332.31   & 1.04 & 1.03  \\
10  & 100k  & 2.62  &	2.63	&  538.3   &	539.45  &	658.17   & 1.22 & 1.22  \\
11  & 100k  & 3.44  &	2.75	&  1133.67 &	1134.51 &	1319.14  & 1.16 & 1.16  \\
12  & 100k  & 3.58  &	2.96	&  2333.34 &	2334.14 &	2729.07  & 1.17 & 1.17  \\
\shline
13  & 10k   & 0.54  &	0.17	&  457.9   &	457.94  &	534.16	 & 1.16 & 1.17  \\
14  & 10k   & 0.49  &	0.39	&  939.19  &	939.24  &	1085.97  & 1.15 & 1.16  \\
15  & 10k   & 1.11  &	0.82	&  1385.85 &	1835.91 &	2155.93  & 1.17 & 1.17  \\
16  & 10k   & 2.96  &	2.1	    &  3710.0  &	3710.51 &   4338.54	 & 1.17 & 1.17  \\
\shline
\end{tabular}
\end{table*}

\subsection{Discussion}

From the Table~\ref{tab:random} and Table~\ref{tab:npnmatch}, we can see that sensitivity signatures show great power in Boolean matching.
Although we can only prove that equal sensitivity signatures are the prerequisites of NPN equivalent instead of a necessary and sufficient condition, this is enough to support fast pruning in Boolean matching.
It can help us quickly prune out the mismatched Boolean functions to reduce the runtime.
For small-size Boolean functions, it is more effective than symmetry signatures.
As for large size functions, symmetry also shows good pruning ability due to the limited total matching times and huge amount NPN classes.

For small-size Boolean functions, the proposed matching algorithm gains a better speedup due to sensitivity signatures that will prune more mismatched Boolean functions earlier.
Among 5 to 8 bits, functions with $7$-variables get the maximum speedup, because almost all mismatched Boolean functions pairs were pruned using sensitivity signatures, but there still a lot could not be detected by symmetry signatures. 
When $n$ increases, symmetry signatures can also prune most of the mismatched functions, so the speedup of our algorithm is relatively small.
There is still a little advantage because we use the fast sensitivity computation of sensitivity described in Algorithm~\ref{algorithm:calculation}, which has less computational complexity than symmetry.

\section{Related Work}\label{sec:related}
The core of Boolean matching is to check whether two Boolean functions belong to the same equivalence class (e.g., NPN).
Except for the group algebraic approach~\cite{slepian1953number}, many mature algorithms have been explored in recent years.
These works can be classified into four categories: 1) canonical form-based algorithms, 2) algorithms using Boolean signatures, 3) spectral analysis methods, and 4) SAT-based methods. 
Because SAT-based methods have little relevance to our work, we only focus on the other three methods.

Canonical form-based matching methods compute some complete and unique~(canonical) forms of the Boolean functions. The idea is that two functions match if and only if their canonical forms are identical.
Burch and Long~\cite{burch1992efficient} introduce a canonical form for matching under input negation and a semi-canonical form for matching under input permutation.
Debnath and Sasao~\cite{debnath2004efficient} introduce a canonical form for solving the general Boolean matching problem. 
Lee {\it et al.}~\cite{lee2018canonicalization} devise a procedure to canonicalize a threshold logic function and check the
equivalence of two threshold logic functions by their canonicalized linear inequalities. 
Huang {\it et al.}~\cite{huang2013fast} detect 
symmetry and higher-order symmetry to construct canonical forms.
The power of this kind of method is best manifested in the technology mapping.

A signature of a Boolean function is a compact representation that exploits some properties from the function.
Zhang et al.~\cite{zhang2019efficient} reduce the search space and improve the matching performance by means of structural signatures, variable symmetry, phase collision check, and variable grouping.
Abdollahi and Pedram~\cite{abdollahi2008symmetry} propose new canonical forms based on signatures.

Spectral analysis methods usually transform Boolean functions into spectral representations.
These spectral representations can also be regarded as signatures.
Moore et al.~\cite{moore2006boolean} and Thornton et al.~\cite{thornton2002logic} use Walsh spectra and Haar spectra to finish the Boolean matching and check the equivalence, respectively.
Spectral analysis methods are usually less practical than other approaches due to the exponential size of the spectra.

All previous works did not consider the sensitivity properties of Boolean functions to develop a fast Boolean matching method.
In this paper, we cooperate sensitivity with canonical form-based methods to complete the Boolean matching. 
We firmly believe that sensitivity is a very important property of Boolean functions and can be combined with other methods to find better Boolean matching algorithms. 
We will further explore these in follow-up works.
In the future, we will apply this method to practical applications to evaluate its performance.
And we will explore more sensitivity properties, such as block sensitivity, and try to propose a new canonical form in Boolean matching based on sensitivities.


\section{Conclusion}\label{sec:conclusion}
This paper introduced Boolean sensitivity as a new series of signatures into Boolean matching and proposed a fast matching algorithm based on sensitivity signatures pruning.
We proved that these sensitivity signatures are equal, which are the prerequisites for the NPN equivalence.
We also developed a fast sensitivity calculation method to compute and compare these signatures.
Sensitivity signatures could be easily integrated into traditional methods and distinguish the mismatched Boolean functions faster.
The experimental results show that sensitivity-related signatures we proposed in this paper can reduce the search space to a very large extent, and perform up to 3x speedup over the state-of-the-art Boolean matching methods.

\section*{Acknowledgment}
This work is partly supported by Zhejiang Provincial Key R\&D program under Grant No.~2020C01052, Beijing Municipal Science and Technology Program under Grant No.~Z201100004220007, National Natural Science Foundation of China (NSFC) under Grant No.~62090021, and Beijing Academy of Artificial Intelligence (BAAI).


\newpage
\bibliographystyle{IEEEtran}
\bibliography{IEEEabrv, npn}

\end{CJK*}
\end{document}